\documentclass{article}
\usepackage{etex}
\reserveinserts{28}
\usepackage{setspace, amsmath, amssymb, amsthm}
\usepackage[colorlinks = true]{hyperref}
\usepackage[margin = 2.5cm]{geometry}
\usepackage[english]{babel}
\usepackage[utf8x]{inputenc}
\usepackage{amsmath}
\usepackage{graphicx}
\usepackage[usenames,dvipsnames]{pstricks}
\usepackage{epsfig}
\usepackage{pst-grad} 
\usepackage{pst-plot} 
\usepackage{epstopdf}
\usepackage{subcaption}
\usepackage{tikz}
\usetikzlibrary{arrows}
\usetikzlibrary{calc}
\usetikzlibrary{fit}
\usepackage{pstricks-add}
\newtheorem{theorem}{Theorem}[section]
\newtheorem{lemma}[theorem]{Lemma}

\newtheorem{corollary}[theorem]{Corollary}

\newtheorem{definition}{Definition}
\newtheorem{remark}{Remark}
\usepackage{tikz}

\DeclareMathOperator{\diag}{diag}

\title{Nonsingular (Vertex-Weighted) Block Graphs\thanks{This work is supported by the Joint NSFC-ISF Research Program (jointly funded by
the National Natural Science Foundation of China and the Israel Science Foundation (Nos.
11561141001, 2219/15), the National Natural Science Foundation of China (No.11531001).}}
 \author{Ranveer Singh\thanks{ Technion-Israel Institute of Technology, Haifa 32000, Israel. Email: \texttt{ranveer@iitj.ac.in} (Corresponding author)}\\
  Cheng Zheng\thanks{Technion-Israel Institute of Technology, Haifa 3200003, Israel. Email: \texttt{cheng.zheng@campus.technion.ac.il}}\\
Naomi Shaked-Monderer\thanks{The Max Stern Yezreel Valley College, Yezreel Valley 1930600, Israel. Email: \texttt{nomi@technion.ac.il}}\\
Abraham Berman\thanks{Technion-Israel Institute of Technology, Haifa 3200003, Israel. Email: \texttt{berman@technion.ac.il}}}

\begin{document}

        \maketitle
\begin{abstract}
A graph $G$ is \emph{nonsingular (singular)} if its adjacency matrix $A(G)$ is nonsingular (singular). In this article, we consider the nonsingularity of block graphs, i.e., graphs in which every block is a clique. Extending the problem, we characterize nonsingular vertex-weighted block graphs in terms of reduced vertex-weighted graphs resulting after successive deletion and contraction of pendant blocks. Special cases where nonsingularity of block graphs may be directly determined are discussed.
\end{abstract}

\section{Introduction}
In 1957 Collatz and Sinogowitz proposed the problem of characterizing nonsingular graphs, i.e, graphs whose adjacency matrix is nonsingular \cite{von1957spektren}. 
This problem is of much interest in various branches of science, in particular  quantum chemistry, H\"{u}ckel molecular orbital theory  \cite{gutman2011nullity, lee1994chemical} and social networks theory \cite{leskovec2010signed}. Significant work was done towards a solution to this problem for special classes of undirected graphs, such as trees, unicyclic and bicyclic graphs \cite{cvetkovic1980spectra, fiorini2005trees, singh2018b, gutman2001nullity, hu2008nullity, singh2017characteristic, nath2007null, bapat2011note,
berman2008upper,
xuezhong2005nullity, sciriha2007characterization, sciriha1998construction}. In particular, a tree is nonsingular if and only if it has a perfect matching  \cite{cvetkovic1972algebraic}. Block graphs are a natural generalization of trees. A \emph{block} in a graph is a maximal connected subgraph with no cut-vertex. A \emph{block graph} is a graph in which each block is a clique (i.e., a complete subgraph), see \cite[p. 15]{west2001introduction}, \cite{bapat2014adjacency}. In this article we study nonsingularity of block graphs.

It turns out that in order to characterize nonsingular block graphs, it is useful to consider vertex-weighted graph.
A \emph{vertex-weighted graph} is a pair $(G, x),$ where $G=(V(G),E(G))$ is a  simple graph with vertex set $V(G)=\{1, \hdots, n\}$, edge set $E(G)$, and $x\in \mathbb{R}^n$ is a vector of vertex weights, $x_i$ is the weight of vertex $i$.  A graph $G$ is the vertex-weighted block  graph $(G, o)$, where $o$ is the zero vector. The adjacency matrix $A(G, x)$ of $(G, x)$
is given  by
$$A(G, x)=A(G)+\diag(x),$$ where $\diag(x)$ is a diagonal matrix whose $i$-th diagonal entry is $x_i$. If $(G,x)$ is a vertex-weighted graph, and
$H$ is a subgraph of $G$, we denote by $x^H$ the restriction of the vector $x$ to the vertices of $H$.
We refer to $(H,x^H)$ as a  subgraph of $(G,x)$, and if $H$ is  a
component of $G$ we refer to $(H,x^H)$ as a  component  of $(G,x)$.

 A vertex-weighted block  graph $(G, x)$ is \emph{nonsingular
(singular)} if $A(G, x)$ is nonsingular (singular).
In Section 2, we give a necessary and sufficient condition for a vertex-weighted block graph to be singular in terms of its reduced graphs resulting after successive contraction and deletion of pendant blocks. We then, in Section 3, present several families of nonsingular block graphs.  In Section 4, we show that replacing edge blocks by paths of even order preserve nonsingularity/singularity. 

The following terms and notations are used in the paper. A graph $G$ is a \emph{coalescence (at the vertex $v$)} of  two disjoint graphs $G_1$ and $G_2$ if it is attained by identifying a vertex $v_1\in V(G_1)$
and a vertex $v_2\in V(G_2)$, merging the two vertices into a single vertex $v$. We use $J,j,O,o,w$ to denote an all-ones matrix, an all-ones column vector, a zero matrix, a zero column vector and a  $(0,1)$-vector of suitable order, respectively.
 The standard basis vectors in $\mathbb{R}^n$ are denoted by $e_1, \dots, e_n$.
 A clique on $n$ vertices is denoted by $K_n$.  If $Q$ is a subgraph of $G$, then $G \setminus Q$ denotes the induced subgraph of $G$ on the vertex subset $V(G)
 \setminus V(Q)$. If $Q$ consists of a
single vertex $v$ we will write $G\setminus v$
for $G\setminus Q$. 
The determinant of a graph $G$ is $\det(G)=\det(A(G))$.
For a nonzero $\alpha \in \mathbb{R}$,  we use in this paper the following notation:
\[\alpha^{1/2}=\begin{cases}
\sqrt{\alpha}& \mbox{if $\alpha>0$}\\
i\sqrt{|\alpha|}&\mbox{if $\alpha<0$}.
\end{cases}\]
For a diagonal matrix $D$ with nonzero real diagonal entries, $D^{1/2}$ and $D^{-1/2}$ are interpreted accordingly.

\section{Characterizing nonsingular vertex-weighted block graphs}
We start with a complete characterization of nonsingular vertex-weighted complete graphs,
and some implications for vertex-weighted graphs that have a  pendant block  which is a clique. Note that  elementary row and column operations do not change the rank of a matrix, and we use
this fact in checking the singularity of $A(G,x)$. In particular,
simultaneous permutations of rows and columns of $A(G,x)$
do  not change the rank, thus in checking whether a
 vertex-weighted block  graph $(G, x)$ is singular or not we may relabel the vertices of $G$, and reorder $x$ accordingly, as convenient.

\begin{theorem}\label{l1}
Let $x\in \mathbb{R}^n$.
\begin{enumerate}
\item If exactly one of $x_1,\hdots, x_n$ is equal to 1, then $(K_n, x)$ is nonsingular.
\item If at least two of $x_1,\hdots, x_n$ are equal to 1, then $(K_n, x)$ is singular.
\item  If $x_i\ne 1$, $i=1,\hdots,n$, let \begin{equation}t(x)=\sum_{i=1}^n\frac{1}{1-x_i}. \label{t(x)}\end{equation}
then
\begin{enumerate}
\item $(K_n, x)$ is nonsingular if and only if $t(x) \neq 1$.
\item if $(K_n, x)$ is singular, then  for any vector $y\in \mathbb{R}^{n+1}$ such that $y_i=x_i$, $i=1,\hdots, n$, and $y_{n+1}\ne 1$,  the graph $(K_{n+1}, y)$ is
    nonsingular.  \end{enumerate}
\item  If $x_i\ne 1$ for $i=1,\hdots, n$  and $(K_n, x)$ is nonsingular, then  any  matrix $M$ of the form
\[M=\begin{bmatrix}
A(K_n, x) & j &O^T\\ j^T & \alpha&w^T\\O&w&B
\end{bmatrix},\]
can be transformed, using elementary row and column operations,  to the following matrix
\[\begin{bmatrix}
A(K_n, x) & o &O^T\\ o^T & \alpha+\gamma (K_n, x)&w^T\\O&w&B
\end{bmatrix},\]
where
 \begin{equation}\label{valueofgamma}
 \gamma (K_n, x)= \begin{cases}
-\frac{t(x)}{t(x)-1} & \mbox{if $x_i\neq 1,~ i=1,\hdots n$},\\
 -1 & \mbox {if exactly one of $x_1, \hdots,x_n $ is equal to 1.}
\end{cases}
 \end{equation}
\end{enumerate}
\end{theorem}

\begin{proof}
Let $D=\diag(j-x)$. Then \[
A(K_n, x)=J-D.
\]

\begin{enumerate}
\item  Without loss of generality, let $x_1=1$. By subtracting the first row from the next $n-1$ rows we get that $A(K_n,x)$ is row-equivalent to the matrix
$$\begin{bmatrix}
1 & 1 & \hdots & 1\\ 0 &  x_2-1  &\ddots &\vdots \\ \vdots & \ddots & \ddots  & 0\\0 &\hdots & 0&  x_n-1
\end{bmatrix},$$ whose determinant $\prod_{i=2}^n(x_i-1)$ is nonzero.
\item In this case two rows (or columns) are equal.
\item \begin{enumerate}
\item  Denote \begin{equation}\label{pdef} p=D^{-1/2}j. \end{equation}  Then
\begin{eqnarray}
J-D&=&D^{1/2}(D^{-1/2}JD^{-1/2}-I)D^{1/2}\nonumber\\
&=& D^{1/2}(D^{-1/2}jj^TD^{-1/2}-I)D^{1/2}\label{d-j}\\
&=& D^{1/2}(pp^T-I)D^{1/2}\nonumber
\end{eqnarray}
 Thus $J-D=A(K_n, x)$ is nonsingular if and only if $pp^T-I$ is nonsingular. Since the eigenvalues of $pp^T$ are $p^Tp$ and $0$, the eigenvalues of the matrix $pp^T-I$ are
 $pp^T-1$ and $-1$.  Thus $pp^T-I$, and $J-D$, are nonsingular if and only if $p^Tp\ne 1$.  As $p_i=({1-x_i})^{-1/2}$,
\[p^Tp=\sum_{i=1}^n\frac{1}{1-x_i}=t(x),\]
and  $(K_n, x)$ is nonsingular if and only if $t(x)\neq 1$.

\item If $t(x)=1$, then
if $y_{n+1}\neq 1$,   $$t(y)=t(x)+\frac{1}{1-y_{n+1}}\neq 1,$$  and $(K_{n+1},y)$ is nonsingular by part 3(a);
and if $y_{n+1}= 1$,  $(K_{n+1},y)$ is nonsingular by part 1 of the  theorem.

\end{enumerate}
\item \begin{enumerate}
\item Let $x_i\neq 1$, $i=1,\hdots n$. For every $p\in \mathbb{R}^n$,
 $$(I-pp^T)(I+spp^T)=I+(s-1-sp^Tp)pp^T.$$ Therefore if $p^Tp\neq 1$,
 $$(I-pp^T)^{-1}=I+\frac{1}{1-p^Tp}pp^T.$$ Thus if $A(K_n,x)=J-D$ is invertible, where
$D=\diag(j-x)$, then by  \eqref{d-j}  above, $$(A(K_n,x))^{-1}=(J-D)^{-1}=-D^{-1/2}\Big(I+\frac{1}{1-p^Tp}pp^T\Big)D^{-1/2}.$$
Hence
 \begin{align*} \label{vog}
j^T(A(K_n,x))^{-1}j &=-j^TD^{-1/2}\Big(I+\frac{1}{1-p^Tp}pp^T\Big)D^{-1/2}j\nonumber\\ &= -p^T\Big(I+\frac{1}{1-p^Tp}pp^T\Big)p=\frac{p^Tp}{p^Tp-1}.
\end{align*}

 Let
\begin{equation*}\label{fjds}
P=\begin{bmatrix}
I & -A(K_n, x)^{-1}j & O\\ o^T & 1 & o^T\\ O^T & o & I
\end{bmatrix}.
\end{equation*}
Then
\begin{equation*}\label{fjde}
P^TMP=\begin{bmatrix}
A(K_n, x) & o &O^T\\ o^T & \alpha+\gamma (K_n, x)&w^T\\O&w&B
\end{bmatrix},
\end{equation*}
where  \begin{equation*} 
\gamma(K_n, x)=-j^TA(K_n, x)^{-1}j=-\frac{p^Tp}{p^Tp-1}=-\frac{t(x)}{t(x)-1}.
\end{equation*}

\item When exactly one of  $x_1,\hdots, x_n$ is equal to 1, we may assume without loss of generality that $x_1=1$. If in
\begin{equation*}
M=\begin{bmatrix}
A(K_n, x) & j &O^T\\ j^T & \alpha& w^T\\O&w&B
\end{bmatrix}
\end{equation*}
we subtract the first column from column $n+1$, and then  the first row from  row $n+1$, we get the following matrix:
\begin{equation*}
\begin{bmatrix}
A(K_n, x) & o &O^T\\ o^T & \alpha+\gamma(K_n, x)&w^T\\O&w&B
\end{bmatrix},
\end{equation*} where $\gamma(K_n, x)=-1.$ \pushQED{\qed} \qedhere
\end{enumerate}
\end{enumerate}
\end{proof}

\begin{remark}\label{remarkgamma}{\rm
In part 4 of Theorem \ref{l1}, note the following special cases for $n\ge 2$:
\begin{enumerate}
\item If $x_i<1$ for every $1\le i\le n$, and $x_i=0$ for at least one $1\le i\le n$, then $t(x) >1.$ Hence in this case $\gamma(K_n, x) < -1$.

\item If $x$ is a zero vector,  $A(K_n, o)^{-1}=-I+\frac{1}{n-1}J$ is a matrix with all diagonal elements equal to $-\frac{n-2}{n-1}$ and all off diagonal elements equal to
    $\frac{1}{n-1}.$ In this case we get that
\begin{equation*}
\gamma(K_n, o)=-\frac{n}{n-1}.
\end{equation*}
\end{enumerate}}
\end{remark}

\begin{remark}\label{rm1.12}
{\rm Part 2 of Theorem \ref{l1} may be generalized: If a vertex-weighted block  graph $(G,x)$ has a block $(B,x^B)$
such that $x_i=x_j=1$ for two non-cut-vertices $i\ne j$, then $(G,x)$ is singular.}
\end{remark}

For a block $(B, x^B)$  of  a vertex-weighted block  graph $(G, x)$, we denote by  $\bar{x}^B$  the sub-vector of $x^B$ consisting of the entries corresponding to the
non-cut-vertices in $(B, x^B)$.
If $x_i\ne 1$ for every non-cut-vertex in $B$,
we define \begin{equation}\label{tau}
\tau_{(G,x)}(B,x^B)=t(\bar{x}^B).
\end{equation}
We simplify the notation to $\tau(B)$ when no confusion may arise.

We now define two operations on $(G, x)$ using its pendant blocks.

\begin{definition}\label{PB-reds}
{\rm
\begin{enumerate}
\item \textbf{PB-deletion.} Let $(B, x^B)$ be a pendant block such that $\bar{x}^B_i\ne 1$ for every $i$, and $\tau(B)=1$.
A \emph{PB-deletion    of $(B, x^B)$}  is the operation of deleting all the vertices of $B$ and the corresponding entries of the
weights vector $x$, yielding a subgraph $(H, x^H)$, where $H=G \setminus B$.

\item \textbf{PB-contraction.} Let $(B, x^B)$ be a pendant block  of $(G,x)$ with a cut-vertex $k$, such that
either exactly one entry in $\bar{x}^B$ is 1, or $\bar{x}^B_i\ne 1$ for every $i$  and $\tau(B)\ne 1$.
A \emph{PB-contraction of $(B,x^B)$} is the operation of merging all
the vertices of $(B, x^B)$ to the cut vertex $k$, deleting the entries of $\bar{x}^B$ from $x$, and adding the weight $\gamma(B,x^B)$ to $x_k$,  where
$$\gamma(B,x^B)=\begin{cases}
-1 & \mbox{if exactly one entry in $\bar{x}^B$ is 1},\\
-\frac{\tau(B)}{\tau(B)-1} & \mbox{if no entry in $\bar{x}^B$ is 1}.
\end{cases}$$
\end{enumerate}}
\end{definition}

Note that when $(G,x)$ is a vertex-weighted block graph, both PB-deletion and PB-contraction generate a vertex-weighted block graph. Also,
PB-deletions may disconnect a connected vertex-weighted block graph, but PB-contractions preserve connectivity.

\begin{lemma}\label{PBdel} Let $(B, x^B)$ be a pendant block of $(G, x)$  such that $x_i\ne 1$ for every non-cut-vertex in $B$,
and $\tau(B)=1.$
Let $(H, x^H)$ be obtained from $(G, x)$ by PB-deletion of $(B, x^B)$. Then $(G, x)$ is singular if and only if $(H, x^H)$ is singular.
\end{lemma}

\begin{proof}  Without loss of generality we may assume that the vertices of $B$ are $\{1, \dots, k\}$, and $k$ is the cut-vertex.
Then \[A(G,x)=\begin{bmatrix}
A_1 & j &O^T\\ j^T & x_k&w^T\\O&w&A_2
\end{bmatrix},\]
where $A_1=A(K_{k-1}, \bar{x}^B)$ and $A_2=A(H,x^H)$.
Any nonzero minor on the first $k$ rows and some $k$ columns, cannot have a zero column,
cannot have more than one column of the form $e_k$, and cannot consist of the first $k-1$ columns and a column of the form $e_k$,
since $A_1$ is singular. Thus every such nonzero minor includes
the $k$-th column, and any  nonzero  minor that does not include all the first $k-1$ columns has a zero complementary minor.
Hence the Laplace expansion of $\det A(G,x)$ along the first $k$ rows yields
$$\det A(G, x)=\det\begin{bmatrix}
A_1 & j\\ j^T & x_k
\end{bmatrix}\det A_2 =\det A(B, x^B) \det A(H, x^H).$$
(see also {\cite[Lemma 2.3]{singh2017characteristic}}.)

By   part 3(b) of Theorem \ref{l1}, $A(B, x^B)$ is nonsingular. Thus $(G,x)$ is nonsingular if and only if $(H,x^H)$ is nonsingular.
\end{proof}

\begin{lemma}\label{PBcont}    Let $(B, x^B)$ be a pendant block of $(G, x)$ such that either
$\bar{x}^B_i\ne 1$ for  $i$ of $B$ and $\tau(B)\ne 1$,
or exactly one entry in $\bar{x}^B$ is 1.
Let $(H, y)$ be obtained from $(G, x)$ by a PB-contraction of $(B, x^B)$. Then $(G, x)$ is singular if and only if $(H, y)$ is singular.
\end{lemma}

\begin{proof}
Without loss of generality we may assume that the vertices of $B$ are $\{1, \dots, k\}$, and $k$ is the cut-vertex.
Then \[A(G,x)=\begin{bmatrix}
A_1 & j &O^T\\ j^T & x_k&w^T\\O&w&A_2
\end{bmatrix},\]
where $A_1=A(K_{k-1}, \bar{x}^B)$.
If either $x_i\ne 1$ for every non-cut-vertex $i$ of $B$ and $\tau(B)\ne 1$,
or exactly one entry in $\bar{x}^B$ is 1,
the matrix $A_1$ is nonsingular by  part 3(a) and part 1 of  Theorem \ref{l1}.  By part 4 of that theorem, $A(G,x)$ is similar to
the matrix
\[\begin{bmatrix}
A_1 & o &O^T\\ o^T & x_k+\gamma&w^T\\O&w&A_2
\end{bmatrix},\]
where
\[\gamma =\begin{cases}
-1 & \mbox{if exactly one entry in $\bar{x}^B$ is 1},\\
-\frac{\tau(B)}{\tau(B)-1} & \mbox{if no entry in $\bar{x}^B$ is 1}.
\end{cases}\]
Hence $(G,x)$ is nonsingular if and only if
\[A(H,y)=\begin{bmatrix}
 x_k+\gamma&w^T\\w&A_2
\end{bmatrix} \]
is nonsingular.
\end{proof}

\begin{remark}\label{genPB}{\rm
Note that PB-deletion and PB-contraction may be used for any vertex-weighted graph $(G,x)$ which
has a pendant block $(B,x^B)$, where $B$ is a clique, and the proper conditions on $x^B$
are satisfied. Lemmas \ref{PBdel} and \ref{PBcont} hold in this case too.}
\end{remark}

\begin{definition}\textbf{Reduced vertex-weighted block graph. }{\rm A vertex-weighted block  graph $(H,y)$ is a \emph{reduced vertex-weighted block  graph} of the vertex-weighted block  $(G,x)$
if it is obtained from $(G,x)$ by a finite number of PB-deletions and PB-contractions.}
\end{definition}

Lemmas \ref{PBdel} and \ref{PBcont} imply that if $(H,y)$ is a reduced vertex-weighted block  graph of $(G,x)$, then $(G,x)$ is nonsingular if and only if
$(H,y)$ is nonsingular. We can now prove the main theorem.

\begin{theorem}\label{thrm1}
A vertex-weighted block  graph $(G,x)$ is singular if and only if there exists a reduced vertex-weighted block  graph $(H,y)$
that has one of the following:
\begin{enumerate}
\item A component $(B,y^B)$, where $B$ is a clique and $y_i\ne 1$ for every vertex $i$ and $\tau(B)=1$
\item A block $(B,y^B)$ for which at least two entries of $\bar{y}^B$ are equal to 1.
\end{enumerate}
\end{theorem}

\begin{proof} If  $(H,y)$ is a reduced vertex-weighted block  graph of $(G,x)$, and $(H,y)$ satisfies 1 or 2, then $(H,y)$ is singular
by part 3(a) of Theorem \ref{l1} or by Remark \ref{rm1.12}, respectively. By Lemmas \ref{PBdel} and \ref{PBcont} this implies that $(G, x)$ is singular.

Now suppose no reduced vertex-weighted block  graph of $(G,x)$ satisfies 1 or 2. Perform PB-deletions
and PB-contractions on $(G,x)$ until a reduced graph $(H,y)$ is obtained, for which no further PB-deletion
or PB-contraction is possible. As $(H,y)$ cannot be further reduced, and does not satisfy 2, it does not have any
pendant blocks. That is, each of its components is of the form $(B,y^B)$, where $B$ is a clique. Since 1 and 2 are
not satisfied, either $y_i=1$ for exactly one vertex $i$ of $B$, or $y_i\ne 1$ for every vertex $i$ of $B$ and $\tau(B)\ne 1$. Hence by Theorem \ref{l1},
each component of $(H,y)$ is nonsingular, and so is $(G, x)$.
 \end{proof}

We conclude the section with two of examples of families of  vertex-weighted block graphs, where nonsingularity may be
easily checked (without actually reducing the vertex-weighted block graph).

\begin{theorem}\label{tauthr}
Let $(G,x)$ be a  vertex-weighted block graph  that satisfies the following two properties:
\begin{enumerate}
\item[(a)] $x_i\ne 1$ for every vertex $i$.
\item[(b)] $x_i<1$ for every  cut-vertex  $i$.
\item[(c)] For every block $(B,x^B)$ of $(G,x)$, $\tau(B)>1$.
\end{enumerate}
Then $(G,x)$ is nonsingular.
\end{theorem}

\begin{proof}
We show that such $(G,x)$ may be reduced by PB-contractions to a vertex-weighted clique satisfying (a) and (b).
Since such a reduced graph is  nonsingular by Theorem \ref{l1}, this will complete the proof.

It suffices to show that if
 $(B,x^B)$ is a pendant block of $(G,x)$ satisfying (a)--(c), then $(B,x^B)$ may be PB-contracted and the resulting
 vertex-weighted block graph will also satisfy (a)--(c).

Let $k$ be the cut vertex of a pendant block $B$ of $G$.
By (a)--(c), this pendant block may be PB-contracted. The resulting vertex-weighted block graph
$(H,y)$ satisfies $y_i=x_i\ne 1$ for every vertex $i$ of $H$ other than $k$, and
$y_k=x_k-\frac{\tau(B)}{\tau(B)-1}$. As $\tau(B)>1$, $y_k=x_k-\frac{\tau(B)}{\tau(B)-1}<x_k<1$.
Also, for every block $(C,y^C)$ of $(H,y)$, if $k$ is not a vertex in $C$, or $k$ is a cut-vertex in $C$, then  clearly $\tau_{(H,y)}(C,y^C)=\tau_{(G,x)}(C,x^C)>1$.
If $k$ is a non-cut-vertex of $C$ in $(H,y)$, $\tau_{(H,y)}(C, y^C)=\tau_{(G,x)}(C,y^C)+\frac{1}{1-y_k}>1$, since $y_k<1$.
\end{proof}

\begin{theorem}\label{vwb31}
Let $(G,x)$ be a  vertex-weighted block graph, that satisfies the following three properties:
\begin{enumerate}
\item[(a)] $x_i<1$ for every  vertex  $i$.
\item[(b)] Each block $B$ of $G$ has at least $3$ vertices.
\item[(c)] For every block $(B,x^B)$ of $(G,x)$, there exists $i$ such that $\bar{x}^B_i=0$.
\end{enumerate}
Then $(G,x)$ is nonsingular.
\end{theorem}

\begin{proof}
Note that if $(G,x)$ consists of a single block satisfying (a)--(d), then $(G,x)$ is nonsingular:
If $G=K_m$, where $m\ge 3$ and, without loss of generality, $x_1=0$,
\[\tau(G)\ge \sum_{i=1}^m\frac{1}{1-x_i}=1+\sum_{i=2}^m\frac{1}{1-x_i}>1\]
by (a), and thus $(G,x)$  is nonsingular by Theorem \ref{l1}.

If $(G,x)$ has a pendant block $(B,x^B)$, this block may be PB-contracted since
 $\tau(B)>1$ by the first part of Remark 1. As in the previous theorem, the resulting $(H,y)$ also satisfies (a)--(c).
Such $(G,x)$ may be reduced by successive PB-contractions to  a single vertex-weighted block satisfying
(a)--(c), and is therefore nonsingular.
\end{proof}

\begin{remark}\label{remtaub31}
{\rm The two families of vertex-weighted block graphs in Theorems \ref{tauthr} and \ref{vwb31}
are not mutually exclusive, but none of these families fully contains the other.

However, a block graph $G=(G,o)$ satisfies the conditions of  Theorem \ref{tauthr}
if and only if each block of $G$ has two non-cut-vertices. A  block graph $G$ satisfies the conditions of  Theorem \ref{vwb31}
if and only if each block of $G$ has at least three vertices, at least one of which is a non-cut-vertex.
That is, the family of block graphs satisfying Theorem \ref{vwb31} contains all the block graphs satisfying Theorem  \ref{tauthr}.

There are nonsingular  block graphs that do not satisfy the requirements in Theorem \ref{tauthr}. An
example of one such graph is given in Figure \ref{notau} (see Theorem \ref{nkimji}).
}\end{remark}

\section{Some classes of nonsingular  block graphs}
In this section we use Theorem \ref{thrm1} to identify some families of nonsingular block graphs.
First we name the graphs discussed at the end of the previous section.

\begin{definition}\label{b31bg}\textbf{$B^3_1$  block  graph. }{\rm
A block graph is a \emph{$B^3_1$  block  graph} if each block has at least three vertices, at least one of which is a non-cut-vertex.}
\end{definition}

From Theorem \ref{remtaub31} and Remark \ref{remtaub31} we deduce the following.

\begin{theorem}\label{B31bgNS}
Every $B^3_1$  block  graph is nonsingular
\end{theorem}

We observe that using Theorem \ref{thrm1} one obtains  a new proof the following known result.

\begin{theorem}
Let a graph $F$ be a forest on $n$ vertices. Then $F$ is nonsingular if and only if it has a perfect matching.
\end{theorem}

\begin{proof}  Let $F$ be a forest, and let $(B,o^B)$ be any pendant edge in $(F,o)$. Then
$\tau(B)=1$  and $(B,o^B)$ may be PB-deleted, yielding a forest $(G,o^G)$.
Note that $F$ has a perfect matching if and only if $G$ has a perfect matching: if the deleted
pendant edge is $\{u,v\}$, with $v$ the
cut-vertex, then in the PB-deletion all edges incident with $v$ are deleted. Thus if $G$ has a perfect matching, adding
$\{u,v\}$ to this matching yields a perfect matching of $F$. And if $F$ has a perfect matching, $\{u,v\}$ has to be
one of the edges in the matching, and removing it yields a perfect matching of $G$.

Given a forest $F$,  reduce $(F,o)$ as much as possible by PB-deletions, until you get a forest $(H,o^H)$ that
has no pendant edges.  Each component of $(H,o)$ is  either an edge, or a singleton. Then $(H,o^H)$
is nonsingular if and only if no component is a singleton, but also $H$ has a perfect matching if and only if no component
 of $H$ is a singleton. By the above, $F$ has a perfect matching if and only if $H$ has a perfect matching, and
 by Theorem \ref{thrm1} $F$ is nonsingular if and only if $H$ is.
 \end{proof}

Next we consider block graphs of a special construction.

\begin{theorem}\label{nkimji}
Let $G$ be a block graph consisting of a block $K_n$, $n\geq 2$, to which at each vertex $i=1,\hdots,n$, $k_i$ blocks of orders $m^i_1,\hdots,m^i_{k_i}$, each greater
than 2 are attached. Then $G$ is nonsingular if and only if $$\sum_{i=1}^n
\frac{1}{1+\sum_{j=1}^{k_i}\Big(\frac{m^i_j-1}{m^i_j-2}\Big)}\neq 1.$$
\end{theorem}

\begin{proof}
Successively perform PB-contraction of each pendant block of $(G,o)$. Then $(G,o)$ is reduced to a vertex-weighted block  graph $(K_n, x)$. By Remark
\ref{remarkgamma}, $x_i= -\sum_{j=1}^{k_i}\Big(\frac{m^i_j-1}{m^i_j-2}\Big)$. As $x_i\neq 1, i=1,\hdots, n$, we get that $$\tau(K_n, x)=\sum_{i=1}^n
\frac{1}{1+\sum_{j=1}^{k_i}\Big(\frac{m^i_j-1}{m^i_j-2}\Big)}.$$ The result follows by Theorem \ref{thrm1}.
\end{proof}

A special case of  Corollary \ref{nkimji}, where the result is simplified is the following.
Let $n\geq 2, m\geq3, k\geq 1$ be three integers. We define a family of block graph using these three integers. Let us  coalesce $k$ pendant $K_m$ blocks at each vertex of
$K_n$. We call the resulting graph  an {\it $(n,m,k)$-block graph}. As an example the $(4,4,2)$-block graph is shown in Figure \ref{r2biblock}. In the case of
 $(n,m,k)$-block graphs the necessary and  sufficient condition for nonsingularity in  Corollary \ref{nkimji} becomes simple:

\begin{corollary} \label{nmkt}
For $n \geq 1, m\geq 3, k\geq 1$, an $(n,m,k)$-block graph is singular if and only if $$\Bigg(\frac{m-1}{m-2}\Bigg)k=n-1.$$
\end{corollary}

Another special case of Theorem \ref{nkimji} is the case that  $n=2$.

 \begin{corollary}\label{2kimji}
Let $G$ be a block graph consisting of a block $K_2$, to which at each of the two vertices some blocks of order greater than 2 each are attached. Then $G$ is nonsingular.
\end{corollary}

 \begin{proof}
This follows from Theorem \ref{nkimji} for $n=2$, as
\[\pushQED{\qed}\sum_{i=1}^2
\frac{1}{1+\sum_{j=1}^{k_i}\Big(\frac{m^i_j-1}{m^i_j-2}\Big)}< \sum_{i=1}^2\frac{1}{2}=1. \qedhere\]
\end{proof}

\begin{figure}
\centering
     \begin{subfigure}[b]{0.35\textwidth}
\begin{tikzpicture}[scale=0.8]

\draw[fill](0,0) circle[radius=0.1];
\draw[fill](1, 0) circle[radius=0.1];
\draw[fill](0, 1) circle[radius=0.1];
\draw[fill](1,1) circle[radius=0.1];

\draw[thick] (0,0)--(1, 0)--(1,1)--(0,1)--cycle;
\draw[thick] (0,0)--(1,1);
\draw[thick] (0,1)--(1,0);

\draw[fill](1,2) circle[radius=0.1];
\draw[thick] (1,1)--(1,2);
\draw[fill](0.7,1.7) circle[radius=0.1];
\draw[fill](1.3,1.7) circle[radius=0.1];
\draw[thick] (1,1)--(1.3, 1.7)--(1,2)--(0.7,1.7)--cycle;
 \draw[thick] (1.3,1.7)--(0.7,1.7);

\draw[fill](2,1) circle[radius=0.1];
\draw[thick] (1,1)--(2,1);
\draw[fill](1.7,1.3)
circle[radius=0.1];
\draw[fill](1.7,0.7) circle[radius=0.1];
\draw[thick] (1,1)--(1.7, 0.7)--(2,1)--(1.7,1.3)--cycle;
 \draw[thick] (1.7,1.3)--(1.7,0.7);

\draw[fill](2,0) circle[radius=0.1];
\draw[thick] (1,0)--(2,0);
\draw[fill](1.7,0.3)
circle[radius=0.1];
\draw[fill](1.7,-0.3) circle[radius=0.1];
\draw[thick] (1,0)--(1.7, -0.3)--(2,0)--(1.7,0.3)--cycle;
 \draw[thick] (1.7,0.3)--(1.7,-0.3);

\draw[fill](1,-1) circle[radius=0.1];
\draw[thick] (1,0)--(1,-1);
\draw[fill](0.7,-0.7) circle[radius=0.1];
\draw[fill](1.3,-0.7) circle[radius=0.1];
\draw[thick] (1,0)--(0.7, -0.7)--(1,-1)--(1.3,-0.7)--cycle;
 \draw[thick] (1.3,-0.7)--(0.7,-0.7);

\draw[fill](0,-1) circle[radius=0.1];
\draw[thick] (0,0)--(0,-1);
\draw[fill](-0.3,-0.7) circle[radius=0.1];
\draw[fill](0.3,-0.7) circle[radius=0.1];
\draw[thick] (0,0)--(-0.3, -0.7)--(0,-1)--(0.3,-0.7)--cycle;
 \draw[thick] (0.3,-0.7)--(-0.3,-0.7);

\draw[fill](-1,0) circle[radius=0.1];
\draw[thick] (0,0)--(-1,0);
\draw[fill](-0.7,0.3)
circle[radius=0.1];
\draw[fill](-0.7,-0.3) circle[radius=0.1];
\draw[thick] (0,0)--(-0.7, -0.3)--(-1,0)--(-0.7,0.3)--cycle;
 \draw[thick] (-0.7,0.3)--(-0.7,-0.3);

\draw[fill](-1,1) circle[radius=0.1];
\draw[thick] (0,1)--(-1,1);
\draw[fill](-0.7,1.3)
circle[radius=0.1];
\draw[fill](-0.7,0.7) circle[radius=0.1];
\draw[thick] (0,1)--(-0.7, 0.7)--(-1,1)--(-0.7,1.3)--cycle;
 \draw[thick] (-0.7,1.3)--(-0.7,0.7);

\draw[fill](0,2) circle[radius=0.1];
\draw[thick] (0,1)--(0,2);
\draw[fill](-0.3,1.7) circle[radius=0.1];
\draw[fill](0.3,1.7) circle[radius=0.1];
\draw[thick] (0,1)--(0.3, 1.7)--(0,2)--(-0.3,1.7)--cycle;
 \draw[thick] (0.3,1.7)--(-0.3,1.7);
\end{tikzpicture}
\caption{A singular $(4,4,2)$-block graph} \label{r2biblock}
    \end{subfigure}  \hspace{10mm} \begin{subfigure}[b]{0.35\textwidth}
\begin{tikzpicture}[scale=0.8]

\draw[fill](0,0) circle[radius=0.1];
\draw[fill](-0.4,.4) circle[radius=0.1];
\draw[fill](0.4,.4) circle[radius=0.1];
\draw[fill](0,.8) circle[radius=0.1];

\draw[thick] (0,0)--(-0.4,.4)--(0.4,.4)--(0,.8)--cycle;
\draw[thick] (0,0)--(0.4,.4);
\draw[thick] (0,0.8)--(-0.4,.4);
\draw[fill](1.5,1.5) circle[radius=0.1];
\draw[thick] (0,0.8)--(1.5,1.5);

\draw[fill](1,1.6) circle[radius=0.1];
\draw[fill](1.1,2.2) circle[radius=0.1];
\draw[thick] (1,1.6)--(1.1,2.2);
\draw[thick] (1,1.6)--(1.5,1.5);
\draw[thick] (1.1,2.2)--(1.5,1.5);

\draw[fill](1.2,1) circle[radius=0.1];
\draw[fill](1.8,1) circle[radius=0.1];
\draw[fill](1.5,0.6) circle[radius=0.1];

\draw[thick] (1.2,1)--(1.5,1.5);
\draw[thick] (1.8,1)--(1.5,1.5);
\draw[thick] (1.5,0.6)--(1.5,1.5);
\draw[thick] (1.2,1)--(1.5,0.6);
\draw[thick] (1.2,1)--(1.8,1);
\draw[thick] (1.8,1)--(1.5,0.6);

\draw[fill](3,3) circle[radius=0.1];
\draw[fill](2.5,3.05) circle[radius=0.1];
\draw[fill](2.8,3.5) circle[radius=0.1];
\draw[thick] (1.5,1.5)--(3,3);
\draw[thick] (2.5,3.05)--(3,3);
\draw[thick] (2.8,3.5)--(3,3);
\draw[thick] (2.5,3.05)--(2.8,3.5);

\draw[fill](4,1.2) circle[radius=0.1];
\draw[fill] (3.6,1.7) circle[radius=0.1];
\draw[fill](4.4,1.6) circle[radius=0.1];
\draw[fill](4.05,2.1) circle[radius=0.1];

\draw[thick] (4,1.2)-- (1.5,1.5);
\draw[thick] (3.6,1.7)--(4.4,1.6);
\draw[thick] (4.05,2.1)-- (4,1.2);
\draw[thick] (3.6,1.7)-- (4,1.2);
\draw[thick] (4.4,1.6)-- (4,1.2);
\draw[thick] (4.05,2.1)-- (4.4,1.6);
\draw[thick] (4.05,2.1)-- (3.6,1.7);
\draw[fill] (3.6,0.8) circle[radius=0.1];
\draw[fill](4.4,0.8) circle[radius=0.1];
\draw[fill](4.05,0.5) circle[radius=0.1];
\draw[thick] (3.6,0.8)--(4.4,0.8);
\draw[thick] (4.05,0.5)-- (4,1.2);
\draw[thick] (3.6,0.8)-- (4,1.2);
\draw[thick] (4.4,0.8)-- (4,1.2);
\draw[thick] (4.05,0.5)-- (4.4,0.8);
\draw[thick] (4.05,0.5)-- (3.6,0.8);

\draw[fill](6.5,1.4) circle[radius=0.1];
\draw[fill](6,1.8) circle[radius=0.1];
\draw[fill](7,1.8) circle[radius=0.1];
\draw[fill](6,1) circle[radius=0.1];
\draw[fill](7,1) circle[radius=0.1];
\draw[thick] (6.5,1.4)-- (6,1.8);
\draw[thick] (6.5,1.4)-- (7,1.8);
\draw[thick] (6.5,1.4)-- (6,1);
\draw[thick] (6.5,1.4)-- (7,1);
\draw[thick] (7,1.8)-- (6,1.8);
\draw[thick] (6.5,1.4)-- (4,1.2);
\draw[thick] (6,1)-- (7,1);

\end{tikzpicture}
\caption{A nonsingular block graph, Theorem \ref{mnktree} \label{ct}}
    \end{subfigure}   \vspace{10mm}

     \begin{subfigure}[b]{0.35\textwidth}
 \begin{tikzpicture}[scale=0.8]

\draw[fill](0,0) circle[radius=0.1];
\draw[fill](1, 0) circle[radius=0.1];
\draw[fill](0, 1) circle[radius=0.1];
\draw[fill](1,1) circle[radius=0.1];

\draw[thick] (0,0)--(1, 0)--(1,1)--(0,1)--cycle;
\draw[thick] (0,0)--(1,1);
\draw[thick] (0,1)--(1,0);

\draw[fill](1,2) circle[radius=0.1];
\draw[thick] (1,1)--(1,2);
\draw[fill](0.7,1.7) circle[radius=0.1];
\draw[fill](1.3,1.7) circle[radius=0.1];
\draw[thick] (1,1)--(1.3, 1.7)--(1,2)--(0.7,1.7)--cycle;
 \draw[thick] (1.3,1.7)--(0.7,1.7);

\draw[fill](2,1) circle[radius=0.1];
\draw[thick] (1,1)--(2,1);
\draw[fill](1.7,1.3)
circle[radius=0.1];
\draw[fill](1.7,0.7) circle[radius=0.1];
\draw[thick] (1,1)--(1.7, 0.7)--(2,1)--(1.7,1.3)--cycle;
 \draw[thick] (1.7,1.3)--(1.7,0.7);

\draw[fill](2,0) circle[radius=0.1];
\draw[thick] (1,0)--(2,0);
\draw[fill](1.7,0.3)
circle[radius=0.1];
\draw[fill](1.7,-0.3) circle[radius=0.1];
\draw[thick] (1,0)--(1.7, -0.3)--(2,0)--(1.7,0.3)--cycle;
 \draw[thick] (1.7,0.3)--(1.7,-0.3);

\draw[fill](1,-1) circle[radius=0.1];
\draw[thick] (1,0)--(1,-1);
\draw[fill](0.7,-0.7) circle[radius=0.1];
\draw[fill](1.3,-0.7) circle[radius=0.1];
\draw[thick] (1,0)--(0.7, -0.7)--(1,-1)--(1.3,-0.7)--cycle;
 \draw[thick] (1.3,-0.7)--(0.7,-0.7);

\draw[fill](0,-1) circle[radius=0.1];
\draw[thick] (0,0)--(0,-1);
\draw[fill](-0.3,-0.7) circle[radius=0.1];
\draw[fill](0.3,-0.7) circle[radius=0.1];
\draw[thick] (0,0)--(-0.3, -0.7)--(0,-1)--(0.3,-0.7)--cycle;
 \draw[thick] (0.3,-0.7)--(-0.3,-0.7);

\draw[fill](-1,0) circle[radius=0.1];
\draw[thick] (0,0)--(-1,0);
\draw[fill](-0.7,0.3)
circle[radius=0.1];
\draw[fill](-0.7,-0.3) circle[radius=0.1];
\draw[thick] (0,0)--(-0.7, -0.3)--(-1,0)--(-0.7,0.3)--cycle;
 \draw[thick] (-0.7,0.3)--(-0.7,-0.3);
 \end{tikzpicture}
  \caption{A $B^3_1$-block graph}\label{b31fig}
    \end{subfigure} \hspace{15mm} \begin{subfigure}[b]{0.4\textwidth}
 \begin{tikzpicture}[scale=0.8]

\draw[fill](0,0) circle[radius=0.1];
\draw[fill](1, 0) circle[radius=0.1];
\draw[fill](0, 1) circle[radius=0.1];
\draw[fill](1,1) circle[radius=0.1];

\draw[thick] (0,0)--(1, 0)--(1,1)--(0,1)--cycle;
\draw[thick] (0,0)--(1,1);
\draw[thick] (0,1)--(1,0);

\draw[fill](1,2) circle[radius=0.1];
\draw[thick] (1,1)--(1,2);
\draw[fill](0.7,1.7) circle[radius=0.1];
\draw[fill](1.3,1.7) circle[radius=0.1];
\draw[thick] (1,1)--(1.3, 1.7)--(1,2)--(0.7,1.7)--cycle;
 \draw[thick] (1.3,1.7)--(0.7,1.7);

\draw[fill](2,1) circle[radius=0.1];
\draw[thick] (1,1)--(2,1);
\draw[fill](1.7,1.3)
circle[radius=0.1];
\draw[fill](1.7,0.7) circle[radius=0.1];
\draw[thick] (1,1)--(1.7, 0.7)--(2,1)--(1.7,1.3)--cycle;
 \draw[thick] (1.7,1.3)--(1.7,0.7);

\draw[fill](2,0) circle[radius=0.1];
\draw[thick] (1,0)--(2,0);
\draw[fill](1.7,0.3)
circle[radius=0.1];
\draw[fill](1.7,-0.3) circle[radius=0.1];
\draw[thick] (1,0)--(1.7, -0.3)--(2,0)--(1.7,0.3)--cycle;
 \draw[thick] (1.7,0.3)--(1.7,-0.3);

\draw[fill](1,-1) circle[radius=0.1];
\draw[thick] (1,0)--(1,-1);
\draw[fill](0.7,-0.7) circle[radius=0.1];
\draw[fill](1.3,-0.7) circle[radius=0.1];
\draw[thick] (1,0)--(0.7, -0.7)--(1,-1)--(1.3,-0.7)--cycle;
 \draw[thick] (1.3,-0.7)--(0.7,-0.7);

\draw[fill](0,-1) circle[radius=0.1];
\draw[thick] (0,0)--(0,-1);
\draw[fill](-0.3,-0.7) circle[radius=0.1];
\draw[fill](0.3,-0.7) circle[radius=0.1];
\draw[thick] (0,0)--(-0.3, -0.7)--(0,-1)--(0.3,-0.7)--cycle;
 \draw[thick] (0.3,-0.7)--(-0.3,-0.7);

\draw[fill](-1,0) circle[radius=0.1];
\draw[thick] (0,0)--(-1,0);
\draw[fill](-0.7,0.3)
circle[radius=0.1];
\draw[fill](-0.7,-0.3) circle[radius=0.1];
\draw[thick] (0,0)--(-0.7, -0.3)--(-1,0)--(-0.7,0.3)--cycle;
 \draw[thick] (-0.7,0.3)--(-0.7,-0.3);

\draw[fill](-1,1) circle[radius=0.1];
\draw[thick] (0,1)--(-1,1);
\draw[fill](-0.7,1.3)
circle[radius=0.1];
\draw[fill](-0.7,0.7) circle[radius=0.1];
\draw[thick] (0,1)--(-0.7, 0.7)--(-1,1)--(-0.7,1.3)--cycle;
 \draw[thick] (-0.7,1.3)--(-0.7,0.7);

 \end{tikzpicture}
  \caption{A nonsingular block graph, Theorem \ref{nkimji} }\label{notau}
    \end{subfigure}
\caption{Examples of block graphs.}
\end{figure}
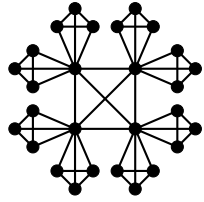
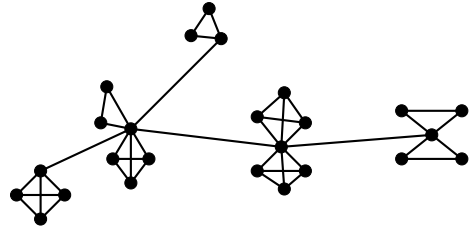
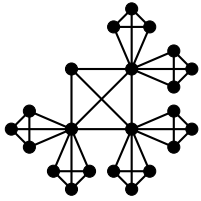
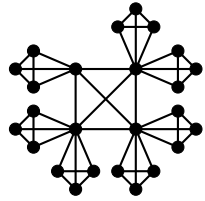

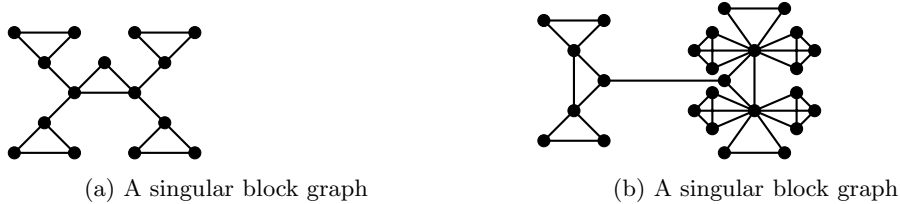
\begin{figure} 
\centering
     \begin{subfigure}[b]{0.35\textwidth}
\begin{tikzpicture}[scale=0.8]

\draw[fill](0,0) circle[radius=0.1];
\draw[fill](1, 0) circle[radius=0.1];
\draw[fill](0.5, 0.5) circle[radius=0.1];

\draw[thick] (0,0)--(1, 0)--(0.5,0.5)--cycle;
\draw[fill](1.5, 0.5) circle[radius=0.1];
\draw[fill](2, 1) circle[radius=0.1];
\draw[fill](1, 1) circle[radius=0.1];
\draw[thick] (1,1)--(2, 1)--(1.5,0.5)--cycle;

\draw[thick] (1,0)--(1.5,0.5);

\draw[fill](-0.5, 0.5) circle[radius=0.1];
\draw[fill](-1, 1) circle[radius=0.1];
\draw[fill](0, 1) circle[radius=0.1];
\draw[thick] (0,1)--(-1, 1)--(-0.5,0.5)--cycle;

\draw[thick] (0,0)--(-0.5,0.5);

\draw[fill](-0.5, -0.5) circle[radius=0.1];
\draw[fill](-1, -1) circle[radius=0.1];
\draw[fill](0, -1) circle[radius=0.1];
\draw[thick] (0,-1)--(-1, -1)--(-0.5,-0.5)--cycle;

\draw[thick] (0,0)--(-0.5,-0.5);

\draw[fill](1.5, -0.5) circle[radius=0.1];
\draw[fill](2, -1) circle[radius=0.1];
\draw[fill](1, -1) circle[radius=0.1];
\draw[thick] (1,-1)--(2, -1)--(1.5,-0.5)--cycle;

\draw[thick] (1,0)--(1.5,-0.5);
\end{tikzpicture}
\caption{A singular block graph} \label{r2biblock}
    \end{subfigure}  \hspace{10mm} \begin{subfigure}[b]{0.35\textwidth}
\begin{tikzpicture}[scale=0.8]

\draw[fill](0,0) circle[radius=0.1];
\draw[fill](2,0) circle[radius=0.1];

\draw[fill](-0.5,-0.5) circle[radius=0.1];
\draw[fill](-0.5,0.5) circle[radius=0.1];

\draw[thick] (0,0)--(-0.5,-0.5)--(-0.5,0.5)--cycle;

\draw[fill](0,1) circle[radius=0.1];
\draw[fill](-1,1) circle[radius=0.1];

\draw[thick] (0,1)--(-0.5,0.5)--(-1,1)--cycle;
\draw[fill](0,-1) circle[radius=0.1];
\draw[fill](-1,-1) circle[radius=0.1];

\draw[thick] (0,-1)--(-0.5,-0.5)--(-1,-1)--cycle;

\draw[thick] (0,0)--(2,0);
\draw[fill](2.5,0.5) circle[radius=0.1];
\draw[fill](2.5,-0.5) circle[radius=0.1];
\draw[thick] (2,0)--(2.5,-0.5)--(2.5,0.5)--cycle;

\draw[fill](2,-1.2) circle[radius=0.1];
\draw[fill](3,-1.2) circle[radius=0.1];
\draw[thick] (2,-1.2)--(3,-1.2)--(2.5,-0.5)--cycle;

\draw[fill](2,1.2) circle[radius=0.1];
\draw[fill](3,1.2) circle[radius=0.1];
\draw[thick] (2,1.2)--(3,1.2)--(2.5,0.5)--cycle;

\draw[fill](3.5,0.5) circle[radius=0.1];
\draw[fill](3.2,0.8) circle[radius=0.1];
\draw[fill](3.2,0.2) circle[radius=0.1];
\draw[thick] (2.5,0.5)--(3.5,0.5)--(3.2,0.8)--(3.2,0.2)--cycle;
\draw[thick] (2.5,0.5)--(3.2,0.8);
\draw[thick] (3.5,0.5)--(3.2,0.2);

\draw[fill](1.5,0.5) circle[radius=0.1];
\draw[fill](1.8,0.8) circle[radius=0.1];
\draw[fill](1.8,0.2) circle[radius=0.1];
\draw[thick] (2.5,0.5)--(1.5,0.5)--(1.8,0.8)--(1.8,0.2)--cycle;
\draw[thick] (2.5,0.5)--(1.8,0.8);
\draw[thick] (1.5,0.5)--(1.8,0.2);

\draw[fill](3.5,-0.5) circle[radius=0.1];
\draw[fill](3.2,-0.2) circle[radius=0.1];
\draw[fill](3.2,-0.8) circle[radius=0.1];
\draw[thick] (2.5,-0.5)--(3.5,-0.5)--(3.2,-0.2)--(3.2,-0.8)--cycle;
\draw[thick] (2.5,-0.5)--(3.2,-0.2);
\draw[thick] (3.5,-0.5)--(3.2,-0.8);

\draw[fill](1.5,-0.5) circle[radius=0.1];
\draw[fill](1.8,-0.2) circle[radius=0.1];
\draw[fill](1.8,-0.8) circle[radius=0.1];
\draw[thick] (2.5,-0.5)--(1.5,-0.5)--(1.8,-0.2)--(1.8,-0.8)--cycle;
\draw[thick] (2.5,-0.5)--(1.8,-0.2);
\draw[thick] (1.5,-0.5)--(1.8,-0.8);

\end{tikzpicture}
\caption{A singular block graph\label{ct}}
    \end{subfigure}
\caption{Examples in support of the necessity of conditions (a) and (b), respectively, in Theorem \ref{treeofb13g}. }\label{fig2}
\end{figure}

Next we consider the following construction.
\begin{definition} \label{treeofBG}\textbf{A tree of block graphs. }{\rm
Let $T$ be a tree on $k$ vertices, and let $G_1, \dots, G_k$ be
block graphs. For every edge $e=\{i,j\}$ of $T$, choose a vertex $u_e$ of $G_i$ and $v_e$ of $G_j$, and connect
$u_e$ and $v_e$ by an edge. The resulting graph $G$ is a block graph, and we call such graph
a {\it tree of $G_1, \dots, G_k$}.
We refer to each of the edges $\{u_e,v_e\}$ in $G$   as a \emph{skeleton edge}, and to the vertices $u_e$ and $v_e$ as \emph{skeleton vertices}.
The graph $G_i$ is considered {\it pendant} in the tree  of $G_1, \dots, G_k$ if the vertex $i$ is pendant in $T$.}
\end{definition}

The first result on a tree of block graphs  generalizes  Corollary \ref{2kimji}.

\begin{theorem}\label{mnktree}
Let $T$ be a tree with $n$ vertices $i=1,\hdots,n$, and let $d(i)$ be the degree of vertex $i$ in $T$. Let $G$ be the graph obtained by coalescing  $k_i$ cliques
$K_{m_1^i},\hdots,K_{m_{k_i}^i}$, each of order at least $3$,  at each vertex $i$ of $T$. If
 $$ \sum_{j=1}^{k_i}\frac{m^i_j-1}{m^i_j-2}>d(i)$$
 for every $i$, then  $G$ is nonsingular.
\end{theorem}

\begin{proof}
By PB-contractions of all pendant blocks in $(G,o)$ we obtain the reduced vertex-weighted tree $(T,x)$, where
$$x_i=\sum_{j=1}^{k_i}-\frac{m^i_j-1}{m^i_j-2}.$$ If $|x_i|> d(i)$ for every $i$, then
$A(T,x)$ is a strictly diagonal dominant matrix, and therefore  nonsingular. The result now follows from Theorem \ref{thrm1}.
\end{proof}

Next consider trees of $B^3_1$ block graphs.

\begin{theorem}\label{treeofb13g}
Let $G$ be a tree of $B^3_1$ block  graphs  $G_1,\dots, G_k$, in which
\begin{enumerate}
\item[(a)] no two skeleton edges share a vertex,
\item[(b)] there is at least one non-cut vertex in any block that has at  $3$ vertices or more.
\end{enumerate}
Then $G$ is nonsingular.
\end{theorem}

\begin{proof}
For such $G$, consider weight vectors $x$ with the following three properties:
  \begin{enumerate}
  \item $x_i<1$ for every $i$.
  \item $x_i=0$ for any skeleton vertex.
  \item For any block $B$ of $G$ with at least three vertices $\bar{x}^B_i=0$ for at
least one vertex $i$.
\end{enumerate}

We show, by induction on $k$, that if $G$ is as in the theorem, and a weight
vector $x$ for $G$ satisfies 1--3, then $(G,x)$ is nonsingular. (As the weight vector $o$ satisfies 1--3,
this will prove the theorem.)

For $k=1$, this holds by Theorem \ref{vwb31}.
Suppose the result holds for any such vertex-weighted tree of $k$ $B^3_1$  block graphs, and let $G$ be a tree of $B^3_1$ block graphs
 $G_1, \dots, G_{k+1}$ that satisfies (a) and (b), and $x$
is a weight vector for $G$, satisfying 1--3.
Without loss of generality, $G_1$ is  pendant in $G$. Let $u\in V(G_1)$ and $v\in V(G\setminus G_1)$ be skeleton vertices.
Then
\[A(G,x)=\begin{bmatrix}
A_1& w_1 & o & O^T\\ w_1^T & 0 & 1 & o^T \\ o^T & 1 & 0 & w_2^T\\ O & o & w_2 & A_2
\end{bmatrix},
\] where $w_1$ and $w_2$ are (0,1)-column vectors, and
$\begin{bmatrix}A_1&w_1\\
w_1^T&0
\end{bmatrix}$
is $A(G_1,x^{G_1})$, and
$\begin{bmatrix}
0&w_2^T\\
w_2&A_2
\end{bmatrix}$
is the adjacency matrix of $(G\setminus G_1, x^{G\setminus G_1})$. As each block graph has at leas two pendant blocks, we
may perform subsequent PB-contractions of blocks in $G_1$, leaving the block containing the skeleton vertex in $G_1$
to last. After these contractions, the  remaining block $(B,b)$   satisfies (a)--(c) of Theorem \ref{vwb31}.
Moreover, $b_i=0$  at least one non-cut-vertex $i$ of $B$. Thus $\tau(B)>1$, and we may contract it also. The adjacency matrix
of the resulting vertex-weighted graph is
\[\begin{bmatrix}
\gamma  & 1 & o^T \\ 1 & 0 & w_2^T\\ o & w_2 &A_2
\end{bmatrix},\]
where $\gamma<-1$ by part 1 of Remark \ref{remarkgamma}.
The pendant edge of this graph has $\tau=\frac{1}{ 1-\gamma}<1$ and may be PB-contracted,
resulting in a weight of $\alpha=-\frac{\tau}{1-\tau}=-\frac{1}{\gamma}<1$ to the vertex $v$.
The resulting  vertex-weighted graph is $(H,y)$, where $H=G\setminus G_1$ is a tree of $G_2, \dots, G_k$, and
$y_i=x_i<1$ for every vertex except  $v$,
whose weight is $\alpha<1$. Note that $v$ is not a skeleton vertex in $H$ (due to the
assumption that in $G$ no two skeleton edges share a vertex).
Thus $y$ satisfies 1--3, and by the induction hypothesis $(H,y)$ is nonsingular. By Theorem \ref{thrm1} so is $(G,x)$
\end{proof}

None of the two conditions (a) and (b)  in Theorem \ref{treeofb13g} may be dropped --- see examples in Figure \ref{fig2}. 

\begin{theorem}\label{taueff+pendedge}
Let $G$ be a block graph, in which each block has at least two non-cut-vertices. Then any graph $G'$
obtained by coalescing edges at some of the cut vertices of $G$ is nonsingular.
\end{theorem}

\begin{proof}
By  PB-deletion of the coalesced pendant edges,
the cut vertices at which they were coalesced are also deleted. The resulting graph is a subgraph of $G$, whose components are $B^3_1$ block graphs,
and is thus nonsingular, implying nonsingularity of $G'$.
\end{proof}

Starting with a graph  like $G'$ of Theorem \ref{taueff+pendedge},
and some nonsingular graphs, we can construct another nonsingular tree of block graphs.

\begin{theorem}\label{W1Ws}
Let $G$ be a block graph, in which each block has at least two non-cut-vertices. Let $G'$ is obtained
as in Theorem \ref{taueff+pendedge} by coalescing edges   at $k$ different
cut vertices $v_1, \dots, v_k$, and let $W_1, \dots, W_s$ be nonsingular block graphs, $s\le k$.
Let $T$ be a star graph $K_{1,s}$. The tree of block graphs of $G',W_1, \dots, W_s$ obtained by choosing $u_i\in V(W_i)$, $i=1, \dots, s$,
and letting the skeleton edges be $\{u_i,v_i\}$, $i=1, \dots, s$, is nonsingular.
\end{theorem}

\begin{proof}
PB-delete each of the $k$ pendant edges. In the resulting graph
each component is either a  $B^3_1$ block graph, or a graph like the one in Theorem \ref{taueff+pendedge}, or one of $W_1, \dots, W_s$.
Thus each component is nonsingular, and so is $G$.
\end{proof}

\begin{figure}
\centering
    \begin{subfigure}[b]{0.30\textwidth}
\begin{tikzpicture} [scale=0.5] [->,>=stealth',shorten >=1pt,auto,node distance=4cm,
                thick,main node/.style={circle,draw,font=\Large\bfseries}]

 \draw  node[draw,circle,fill=black,scale=0.35] (1) at (0,0) {};
\draw  node[draw,circle,fill=black,scale=0.35] (2) at (2,0) {};
\draw  node[draw,circle,fill=black,scale=0.35] (3) at (2,2) {};

\draw  node[draw,circle,fill=black,scale=0.35] (9) at (5,2) {};

\draw  node[draw,circle,fill=black,scale=0.35] (4) at (0,2) {};

\draw  node[draw,circle,fill=black,scale=0.30] (7) at (.75,-1.5) {};
\draw  node[draw,circle,fill=black,scale=0.30] (8) at (-0.75,-1.5) {};

\draw
 node[draw,circle,fill=black,scale=0.35] (6) at (2.75,3.5) {};
\draw  node[draw,circle,fill=black,scale=0.35] (5) at (1.25,3.5) {};

\draw
 node[draw,circle,fill=black,scale=0.30] (10) at (5.75,3.5) {};
\draw  node[draw,circle,fill=black,scale=0.30] (11) at (4.25,3.5) {};
\draw  node[draw,circle,fill=black,scale=0.30] (21) at (3.50,0) {};

\draw  node[draw,circle,fill=black,scale=0.30] (22) at (4.50,.75) {};
\draw  node[draw,circle,fill=black,scale=0.30] (20) at (4.50,-0.75) {};

\draw  node[draw,circle,fill=black,scale=0.30] (12) at (2.75,-2.5) {};
\draw  node[draw,circle,fill=black,scale=0.30] (13) at (4.75,-2.5) {};
\draw  node[draw,circle,fill=black,scale=0.30] (16) at (6.0,-1.5) {};
\draw  node[draw,circle,fill=black,scale=0.30] (17) at (6.0,-3.5) {};

\draw  node[draw,circle,fill=black,scale=0.30] (14) at (2.75,-4.5) {};
\draw  node[draw,circle,fill=black,scale=0.30] (15) at (4.75,-4.5) {};

\draw  node[draw,circle,fill=black,scale=0.30] (18) at (1.25,-2.5) {};
\draw  node[draw,circle,fill=black,scale=0.30] (19) at (1.25,-3.5) {};

\tikzset{edge/.style = {- = latex'}}
\draw[edge] (1) to (2);
 \draw[edge] (1) to (3);
 \draw[edge] (1) to (4);
 \draw[edge] (3) to (2);
 \draw[edge] (4) to (2);
 \draw[edge] (3) to (4);

 \draw[edge] (1) to (7);

 \draw[edge] (1) to (8);

 \draw[edge] (7) to (8);

\draw[edge] (3) to (5);

 \draw[edge] (3) to (6);

 \draw[edge] (5) to (6);

 \draw[edge] (11) to (10);
 \draw[edge] (10) to (9);
 \draw[edge] (11) to (9);

 \draw[edge] (21) to (22);
 \draw[edge] (20) to (22);
  \draw[edge] (21) to (20);

 \draw[edge] (12) to (13);
 \draw[edge] (12) to (14);
 \draw[edge] (12) to (15);
 \draw[edge] (13) to (14);
 \draw[edge] (13) to (15);
 \draw[edge] (14) to (15);
  \draw[edge] (13) to (16);
  \draw[edge] (13) to (17);
  \draw[edge] (16) to (17);
 \draw[edge] (18) to (12);
  \draw[edge] (19) to (12);
  \draw[edge] (18) to (19);
  \tikzset{edge/.style = {thick = latex'}}
  \draw[edge] (11) to (6);
  \draw[edge] (12) to (7);
 \draw[edge] (21) to (2);
\end{tikzpicture}
     \caption{A tree of $B^3_1$-block graphs}\label{b3t}
    \end{subfigure} \begin{subfigure}[b]{0.3\textwidth}
\begin{tikzpicture} [scale=0.5] [every node/.style={draw,shape=circle,fill=blue}]


 \draw node[draw,circle,fill=black, scale=0.35] (1) at (0,0) {};
\draw  node[draw,circle,fill=black, scale=0.35] (2) at (2,0) {};
\draw  node[draw,circle,fill=black,scale=0.35] (3) at (2,2) {};
\draw  node[draw,circle,fill=black,scale=0.30] (12) at (3.5,3.5) {};
\draw  node[draw,circle,fill=black,scale=0.30] (13) at (4,1.5) {};
\draw  node[draw,circle,fill=black,scale=0.30] (14) at (1.5,3.5) {};

\draw  node[draw,circle,fill=black,scale=0.30] (15) at (-0.5,-1.5) {};
\draw  node[draw,circle,fill=black,scale=0.30] (16) at (0.7,-2.5) {};
\draw  node[draw,circle,fill=black,scale=0.30] (17) at (-1,1.75) {};

\draw  node[draw,circle,fill=black,scale=0.35] (4) at (0,2) {};
\draw  node[draw,circle,fill=black,scale=0.30] (5) at (2.75,-2.5) {};
\draw  node[draw,circle,fill=black,scale=0.30] (6) at (4.75,-2.5) {};
\draw  node[draw,circle,fill=black,scale=0.30] (9) at (6.0,-1.5) {};
\draw  node[draw,circle,fill=black,scale=0.30] (10) at (6.0,-3.5) {};

\draw  node[draw,circle,fill=black,scale=0.30] (7) at (2.75,-4.5) {};
\draw  node[draw,circle,fill=black,scale=0.30] (8) at (4.75,-4.5) {};

\draw  node[draw,circle,fill=black,scale=0.30] (20) at (5.5,-4.5) {};

\draw  node[draw,circle,fill=black,scale=0.30] (11) at (3.5,-1) {};
\draw  node[draw,circle,fill=black,scale=0.30] (21) at (5.25,-1) {};

\draw  node[draw,circle,fill=black,scale=0.30] (18) at (5,1) {};

\draw  node[draw,circle,fill=black,scale=0.30] (19) at (7,1) {};

\tikzset{edge/.style = {- = latex'}}

\draw[edge] (11) to (18);
\draw[edge] (11) to (21);
\draw[edge] (19) to (18);
 \draw[edge] (1) to (3);
 \draw[edge] (1) to (4);
 \draw[edge] (1) to (2);
 \draw[edge] (3) to (2);
 \draw[edge] (4) to (2);
 \draw[edge] (3) to (4);
\draw[edge] (6) to (20);
\draw[edge] (5) to (6);
 \draw[edge] (5) to (7);
 \draw[edge] (5) to (8);
 \draw[edge] (6) to (7);
 \draw[edge] (6) to (8);
 \draw[edge] (7) to (8);
 \draw[edge] (6) to (9);
 \draw[edge] (6) to (10);
 \draw[edge] (9) to (10);
 \draw[edge] (6) to (11);
\draw[edge] (3) to (12);
 \draw[edge] (3) to (13);
 \draw[edge] (13) to (12);
 \draw[edge] (14) to (3);

\draw[edge] (1) to (5);
 \draw[edge] (1) to (15);
\draw[edge] (1) to (16);
 \draw[edge] (5) to (16);
 \draw[edge] (5) to (15);
 \draw[edge] (15) to (16);
 \draw[edge] (1) to (17);
\end{tikzpicture}
     \caption{A nonsingular block graph}\label{p1g}
    \end{subfigure}
\caption{} \label{r2fig}
\end{figure}

\section{Replacing edge blocks by even order paths}
We prove here some results on the determinant of a graph obtained
by coalescing two graphs, or combining them by a bridge. These
results will imply ways to construct more nonsingular block graphs from
known block graphs. 

Most of the results in this section are based on \cite[Lemma 2.3]{singh2017characteristic}, restated here for simple graphs with  no vertex weights.
In this lemma, $\phi(G)=\det(A(G)-\lambda I)$ denotes the characteristic polynomial of the graph $G$.

\begin{lemma}\label{SB17} \cite{singh2017characteristic}
Let $G$ be a coalescence of $G_1$ and $G_2$ at a vertex $v$. Then
\[\phi(G)=\phi(G_1)\times \phi(G\setminus G_1)+\phi(G_1\setminus v)\times \phi(G\setminus (G_1\setminus v)+\lambda\times \phi(G_1\setminus v)\times\phi(G\setminus
G_1),\]
\end{lemma}

Using this lemma, we deduce the following.

\begin{lemma} \label{supfirst}
If $G$ is a coalescence of $G_1$ and $G_2$ at a vertex $v$, then $$\det(G)=\det(G_1)\det(G_2\setminus v)+\det(G_1\setminus v)\det(G_2).$$
\end{lemma}

\begin{proof}
Obtain $\det(G)$ by substituting $\lambda=0$ in $\phi(G)$ in Lemma \ref{SB17}. This yields
\[\pushQED{\qed}\det(G)=\det(G_1)\det(G_2\setminus v)+\det(G_1\setminus v)\det(G_2). \qedhere\]
\end{proof}

\begin{corollary}\label{pe}
If a graph $G$ has a pendant edge $\{u,v\}$ with $v$ the cut vertex, then $\det(G\setminus v)=-\det(G)$.
\end{corollary}

\begin{proof}
In this case, $G$ is the coalescence of $G_1=G\setminus v$ and $G_2$ consisting of the edge $\{u,v\}$.
It is easy to see that $\det(G_2)=-2$ and $\det(G_2\setminus v)=0$, the result follows.
\end{proof}

\begin{corollary}\label{SL}
A coalescence of any two singular graphs is singular.
\end{corollary}

\begin{proof}
Let $G$ be coalescence of singular graphs $G_1$ and $G_2.$  As $\det(G_1)=\det(G_2)=0$, $\det(G)=0.$
\end{proof}

Note that a coalescence of nonsingular graphs may be singular: e.g., the coalescence of two edges results in a singular tree.
More generally, we have the following corollary of Lemma \ref{supfirst}.

\begin{corollary}\label{2penedges}
If $G$ is any graph, and two pendant edges are coalesced with it at the same vertex $v$, then
the resulting  graph $G'$ is singular.
\end{corollary}

\begin{proof}
In Lemma \ref{supfirst} let $G_1$ be  the coalescence of one of the pendant edges with $G$, and $G_2$ the second pendant edge.
Then $G_2\setminus v$ is a singleton, and $G_1\setminus v$ has a singleton component, thus $\det(G_2\setminus v)=\det(G_1\setminus v)=0$,
implying that $\det(G')=0$.
\end{proof}

Another way to combine two  graphs $G_1$ and $G_2$ into a larger  graph, is be adding  an edge between a vertex of $G_1$
and a vertex of $G_2$. From Lemma \ref{SL} we get the following.

\begin{lemma}\label{bridge}
Let $G_1$ and $G_2$ be two  disjoint graphs. If we add an edge $\{v_1, v_2\}$, where   $v_1 \in V(G_1)$ and $v_2 \in V(G_2)$, then the
resulting graph is singular if and only if $$\det(G_1)\det(G_2)=\det(G_1 \setminus v_1)\det(G_2 \setminus v_2).$$
\end{lemma}

\begin{proof}
Let $G$ be the resulting graph. Let  $e_{v_1v_2}$ denote the graph consisting of one edge between the vertices $v_1$ and $v_2$ and $G'$ the graph, which is the coalescence of
$e_{v_1v_2}$ and the graph
$G_2$. Note that $G'\setminus v_1=G_2$, $\det(e_{v_1v_2})=-1$ and $\det(v_1)=0$.

Then using Corollary \ref{SL} twice, first for $G$ with the cut vertex $v_1$, and then for $G'$ and the cut vertex $v_2$, we get that
\begin{align*}\det(G)&=\det(G_1)\det(G_2)+\det(G_1\setminus v_1)\det(G')\\
&=\det(G_1)\det(G_2)+\det(G_1\setminus v_1) (\det(e_{v_1v_2}) \det(G_2\setminus v_2)+\det(v_1)\det(G_2))\\
&=\det(G_1)\det(G_2)-\det(G_1\setminus v_1)\det(G_2\setminus v_2). \end{align*}
If $G_1\setminus v_1$ or $G_2\setminus v_2$ is a null graph then the determinant by convention is equal to 1.  
 Thus $G$ is nonsingular if and only
if
\[\pushQED{\qed} \det(G_1)\det(G_2)=\det(G_1\setminus v_1)\det(G_2\setminus v_2).\qedhere\]
\end{proof}

\begin{lemma}\label{G1-G2}
Let $G_1$ and $G_2$ be two graphs. Let $G$ be the graph obtained by adding a path $P$ of order $k$ between a vertex $v_1$ of $G_1$ and a vertex $v_2$ of $G_2$.
Then
\begin{enumerate}
\item If the order $k$ of $P$ is odd, $G$ is nonsingular if the coalescence of $G_1$ and $G_2$ by identifying $v_1$ and $v_2$ is nonsingular.
\item If the order $k$ of $P$ is even, $G$ is nonsingular if the graph $G'$ obtained by connecting $v_1$ and $v_2$ by a single edge is nonsingular.
\end{enumerate}
\end{lemma}

\begin{proof}
For  given graphs $G_1$ and $G_2$, let us denote by $G^{(k)}$ the graph obtained by adding a path $P$ between the vertex $v_1$ of $G_1$ and the vertex $v_2$ of $G_2$.
It suffices to show that for every $k\ge 3$, $G^{(k)}$ is nonsingular if and only if $G^{(k-2)}$ is nonsingular.
Let $k\ge 3$. Choose a vertex $v$ on the path of order $k$ between $v_1$ and $v_2$, whose distance from each of the two end vertices is at least 1.
Let $P'$ be the part of the path $P$ connecting $v_1$ to $v$ (including), $P''$ the part of $P$ connecting $v$ and $v_2$.
Let $v_1'$ be the neighbor of $v$ in $P'$, and $v_2'$ the neighbor of $v$ in $P''$.
Finally, let $G_1'=G_1\cup P'$ and $G_2'=G_2\cup P''$, $G_1''=G_1'\setminus v$ and $G_2''=G_2'\setminus v$. Note that the coalescence of $G_1''$ and $G_2''$
by identifying $v_1'$ and $v_2'$ results in a $G^{(k-2)}$.
By Lemma \ref{supfirst},
\begin{align*}\det(G^{(k)})&=\det(G_1')\det(G_2'\setminus v)+\det(G_1'\setminus v)\det(G_2')\\
&=\det(G_1')\det(G_2'')+\det(G_1'')\det(G_2').
\end{align*}

By  Corollary \ref{pe},
\[\det(G_1''\setminus v'_1)=-\det(G_1')~~\text{ and }~~\det(G_2''\setminus v'_2)=-\det(G_2').\]
Combining that with Lemma  \ref{supfirst} applied to the coalescence of $G_1''$ and $G_2''$ yields
\begin{align*}\det(G^{(k-2)})&=\det(G_1'')\det(G_2''\setminus v'_2)+\det(G_1''\setminus v'_1)\det(G_2'')\\
&=-\det(G_1'')\det(G_2')- \det(G_1')\det(G_2'').
\end{align*}
That is, $\det(G^{(k)})=-\det(G^{(k-2)})$.
\end{proof}

\begin{lemma}\label{NSpentree}
Let $G$ be any graph with a pendant edge $\{u,v\}$, where $v$ is the cut vertex. Let $G'$ be obtained by coalescing $G\setminus u$ with a nonsingular tree $T$ at
the vertex $v$.
Then $G'$ is nonsingular if and only if $G$ is nonsingular.
\end{lemma}

\begin{proof}
Let $v$ be the coalescence vertex. Successively PB-delete  pendant edges of $T$, until
exactly a pendant edge at the vertex $v$ is left. This is possible, since in all
the steps up to the last, there are at least two pendant edges, one with both ends
different than $v$.
\end{proof}

Coalescing a block graph with a tree, and combining block graphs by coalescence or by an edge
yields block graphs. Thus the  results of this section imply the following for block graphs.

\begin{remark}{\rm
Let $G$  be a block graph.
\begin{itemize}
\item If $G$ has two pendant edges at the same cut vertex, then $G$ is singular (Corollary \ref{2penedges}).
\item The coalescence of two singular block graphs is singular (Corollary \ref{SL}).
\item If $G$ has a block, which is an edge $e$, then replacing this edge by a path of even order
results in a block graph $G'$, which is nonsingular if and only if $G$ is nonsingular (by Lemma \ref{G1-G2}, if
$e$ is a bridge, or Lemma \ref{NSpentree}, if $e$ is a pendant edge.)
In particular, this holds for any  tree of block graphs may be thus extended without affecting its singularity/nonsingularity,
and for the graphs in Theorems \ref{taueff+pendedge} and \ref{W1Ws}.
\item Pendant edges  the pendant edges may also be replaced by nonsingular trees without affecting singularity/nonsingularity
(e.g., in Theorems \ref{taueff+pendedge} and \ref{W1Ws}, see Figure \ref{p1g}).
\end{itemize}}
\end{remark}

\bibliographystyle{plain}
\bibliography{RN}

\end{document}